\documentclass[12pt, oneside]{amsart}   	% use "amsart" instead of "article" for AMSLaTeX format
\usepackage{geometry}                		% See geometry.pdf to learn the layout options. There are lots.
\usepackage{graphicx}				% Use pdf, png, jpg, or eps§ with pdflatex; use eps in DVI mode
\usepackage{amscd}								% TeX will automatically convert eps --> pdf in pdflatex		
\usepackage{amssymb}

\title{Generalizing Birkhoff}
\date{}

\newcommand{\eee}{\mathbf{e}}
\newcommand{\xxx}{\mathbf{x}}
\newcommand{\aaa}{\mathbf{a}}
\newcommand{\bbb}{\mathbf{b}}

\newtheorem{thm}{Theorem}
\newtheorem{lem}{Lemma}

\newtheorem{exam}{Example}
\newtheorem{cor}{Corollary}
\newtheorem{defin}{Definition}

\begin{document}

\begin{center}

Joel L. Weiner

Generalizing Birkhoff

University of Hawaii at Manoa, Department of Mathematics\\
Honolulu, HI USA

joel@math.hawaii.edu

\vspace{.5in}

{\bf  Abstract}

\end{center}

We generalize Birkhoff's Theorem in the following fashion.  We find necessary and sufficient conditions for any spherically symmetric space-time to be static in terms of the eigenvalues of the stress-energy tensor.  In particular, we generalize the Tolman-Oppenheimer-Volkoff equation and prove that Birkhoff's theorem holds under the weaker hypothesis of no pressure (with respect to an appropriate frame.)  We provide equations that show how the coefficients of the metric relate to the eigenvalues of the stress-energy tensor.    These involve integrals that are simple functions of those eigenvalues.  We also determine among all static spherically symmetric space-times those that are asymptotically flat.  A few examples are presented taking advantage of the results.  The calculations are done by viewing the space-times as warped products and the computations are done using Cartan's moving frames approach.

\vspace{.5in}

\begin{center}

\vspace{.5in}

\end{center}
\newpage

\maketitle

\section{Introduction}
	Birkhoff's Theorem deals with spherically symmetric space-times.  It asserts that if the stress-energy tensor ${\bf T} = \mathbf{0}$, then such a space-time is static.  This further implies the existence of additional symmetries and, if it makes sense to talk about space at infinity, that the space-time is asymptotically flat.  The most famous example of a  spherically symmetric metric is the Schwarzschild metric, which has a trivial stress-energy tensor and, of course, is static.  There are other examples of  spherically symmetric space-time metrics that are static such as the de Sitter space-time within an appropriately restricted domain and the Reissner-Nordstr\"om space-time.  The former is not asymptotically flat but the latter is.

   There have been many generalizations of Birkhoff's Theorem, most from the following viewpoint:   ``Given a space-time acted upon by a group of symmetries,  under what conditions do additional symmetries exist?"  A number of papers, for example  Bokhari et al. \cite{b} and Shanarii et al. \cite{sh}, look for necessary conditions on static spherically symmetric space-times in order that there exist additional symmetries.   Others, such as Szenthe \cite{sz} look for sufficient conditions for the existence of a Killing vector field orthogonal to the orbits of the action of a given symmetry group.
In fact, the paper by Szenthe gives a thorough review of the literature on that approach to the generalization of Birkhoff's Theorem.

Approaches such these just mentioned focus on the metric.
Other approaches to generalizing Birkhoff look for conditions on the stress-energy tensor ${\bf T}$ of a spherically symmetric space-time in order for it to be static.  Wald \cite{w}, pp.125--128, examines that situation when the stress-energy tensor is that of a perfect fluid and derives the {\em Tolman-Oppenheimer-Volkoff} (TOV) equation.  Moreover, in a paper by Cattoen et al. \cite{c} an anisotropic version of the TOV equation is derived.  These TOV equations arise as necessary conditions because, in both cases, it is assumed from the outset that the metric is static.

Our goal is to find \emph{necessary and sufficient} conditions on ${\bf T}$ for a spherically symmetric space-time to be static and show exactly how such metrics relate to ${\bf T}$.  Then, among those space-times that are static, we further determine which are asymptotically flat.  The paper is organized as follows:  In Section 2 we introduce the moving frame notation we use throughout the paper and present some preliminary calculations; in particular, we get formulas for the Ricci forms associated to a moving frame.  In Section 3, we use the formulas for the Ricci forms to prove the standard Birkhoff Theorem in order to familiarize the reader to the techniques about to be employed in the generalizations.  Section 4 is devoted to our first generalization where we consider stress-energy transformations which have at most 2 distinct eigenvalues.  Section 5 considers the most general situation where the stress-energy transformation may have up to 3 distinct eigenvalues.  It is in this section that we obtain a generalization of the TOV equation and also prove, as a corollary of the main theorem, the following result:  If a spherically symmetric space-time has no pressures with respect to what we call the Birkhoff framing, then it has no energy density with respect to that framing, and thus is a vacuum.  Because of this result, the original Birkhoff Theorem holds under a weaker assumption; this is a second corollary.  A third corollary to the main theorem asserts that spherically symmetric dust cannot be static.  Finally, Section 6 is devoted to determining which  static  spherically symmetric space-times are asymptotically flat.

Lastly, we recommend the paper by Szenthe \cite{sz} for a through presentation of differential topological aspects of the problem under consideration and hence as a good jumping off point for what is to follow.

\section{Preliminary Calculations}

Let us denote by $X$ a space-time which is spherically symmetric.  As a consequence of the spherical symmetry, there exists a fiber bundle $\beta: X \to B$ whose fibers are the 2-spheres that are the orbits of the isometry group that accounts for the spherically symmetric structure.  We may assume, without loss of generality that $B$ is diffeomorphic to an open ball and thus $X = B \times \mathbf{S}^2$.  This is certainly true locally.  Moreover, due to the fact that the symmetries are isometries and the orbits are spacial, the metric $\bar{\mathbf{g}}$ on $X$ projects to a nonsingular metric $\mathbf{g}$ on $B$ which has index 1.

Through each event $\xxx \in X$ there is a round 2-sphere of curvature $\frac{1}{r^2}$.  This determines a function $r : X \to \mathbb{R}$ which we assume is regular; this function projects to a regular function on $B$ which we also denote by $r$. We choose $r$ to be a coordinate function on $B$.  Necessarily there exists a regular function $t$ on $B$ such that $dr$ and $dt$ are orthogonal with respect to $\mathbf{g}$.  We then have a coordinate system $t,r$ on $B$ and we finally assume that the image of $B$ under $(t,r): B \to \mathbb{R}^2$ is a rectangle $(t_0, t_1) \times (r_0 , r_1)$, where $0 < t_0 < t_1 \leq \infty$ and   $0 < r_0 < r_1 \leq \infty$.  Again this is always true locally, so we have not imposed any severe restrictions by these assumptions. We finally observe that $X$ is a warped product; in fact
\[
X = B \times_r \mathbf{S}^2,
\]
in the notation of O'Neill \cite{o} .  Thus, if we let $\mathbf{h}$ denote the metric on the round 2-sphere of radius 1, 
\[
\bar{\mathbf{g}} = \mathbf{g} + r^2 \mathbf{h}.
\]
We intend to stay faithful to O'Neill's terminology throughout this paper.

We will let $\bar{\bf{R}}$ denote the Ricci curvature tensor on $X$ viewed as a symmetric bilinear form.  This notation is convenient and will not lead to any confusion since we can deal with the Riemann curvature tensor without introducing a symbol to represent it.
Will we use $\bf{R}$ to denote the Ricci tensor for $B$, and, of course, the Ricci tensor for $\mathbb{S}^2$ is $\bf{h}$.  Our first goal is to obtain formulas for  $\bar{\bf{R}}$ since that is intimately connected to $\bf{T}$ by means of Einstein's equation.  
 
We will be working with moving frames.  Thus we introduce all of the following, all of which are defined locally.  Let $\aaa_\alpha, \alpha = 0, 1$, be an  orthonormal frame field on $B$ with $\aaa_0$ time-like.  We denote its dual frame field by $\theta^\alpha$.  Let $\bbb_\mu, \mu = 2,3$, be an orthonormal frame field on $\mathbb{S}^2$ with dual frame field $\psi^\mu$.
 Then, there exist connection forms ${\theta^\alpha}\!_\beta$  on $B$ such that
\[
d\theta^\alpha = - {\theta^\alpha}\!_\beta \wedge \theta^\beta.
\]
If ${\Theta^\alpha}\!_\beta$ are the curvature forms associated with $\theta^\alpha$, then 
\[
{\Theta^\alpha}\!_\beta = d{\theta^\alpha}\!_\beta. 
\]
In addition, there exist connection forms ${\psi^\mu}\!_\nu$ on $\mathbb{S}^2$ such that
 \[
 d\psi^\mu = - {\psi^\mu}\!_\nu \wedge \psi^\nu
 \]
 and curvature forms ${\Psi^\mu}\!_\nu$, where
 \[
 {\Psi^\mu}\!_\nu = d{\psi^\mu}\!_\nu  = \psi^\mu \wedge \psi_\nu.
 \]
 
 On $X$, we also  let $\aaa_\alpha$ and $\bbb_\mu$ stand for  the horizontal and vertical vector fields that project to $\aaa_\alpha$ on $B$ and $\bbb_\mu$ on $\mathbb{S}^2$, respectively.  In addition, the pullbacks of all the forms introduced above to $X$ by means of projection mappings onto $B$ and $\mathbb{S}^2$ will be denoted by the same symbols.  Define vector fields $\bar{\aaa}_\alpha$ and $\bar{\bbb}_\mu$ on $X$ by
 \[
 \bar{\aaa}_\alpha = \aaa_\alpha \quad {\rm and } \quad \bar{\bbb}_\mu = \frac{1}{r}\bbb_\mu.
 \]
 Then $\bar{\aaa}_\alpha, \bar{\bbb}_\mu$ is an orthonormal frame on $X$ with dual frame $\bar{\theta}^\alpha, \bar{\psi}^\mu$ where
 \[
 \bar{\theta}^\alpha = \theta^\alpha \quad {\rm and} \quad \bar{\psi}^\mu = r \psi^\mu.
 \]
 
 We introduce connection forms ${\bar{\theta}^\alpha}\!_\beta, {\bar{\psi}^\mu}\!_\nu, {\bar{\omega}^\alpha}\!_\mu$ and ${\bar{\omega}^\mu}\!_\alpha$ where
 \begin{eqnarray*}
 d\bar{\theta}^\alpha & = & - {\bar{\theta}^\alpha}\!_\beta \wedge \bar{\theta}^\beta - {\bar{\omega}^\alpha}\!_\mu \wedge \bar{\psi}^\mu,\\
 d\bar{\psi}^\mu & = & - {\bar{\omega}^\mu}\!_\alpha \wedge \bar{\theta}^\alpha - {\bar{\psi}^\mu}\!_\nu \wedge \bar{\psi}^\nu
 \end{eqnarray*}
 and note that ${\bar{\omega}^\alpha}\!_\mu = \pm {\bar{\omega}^\mu}\!_\alpha$, the sign depending on the metrics $\bf{g} = \theta_\alpha \theta^\alpha$ on $B$ and $\bf{h} = \psi_\mu\psi^\mu$ on $\mathbb{S}^2$.  
 We now compute ${\bar{\theta}^\alpha}\!_\beta, {\bar{\psi}^\mu}\!_\nu,
 {\bar{\omega}^\mu}\!_\alpha$ 
 and ${\bar{\omega}^\alpha}\!_\mu$.  We note that
 \[
 d\bar{\theta}^\alpha = d\theta^\alpha = - {\theta^\alpha}\!_\beta \wedge \theta^\beta = - {\theta^\alpha}\!_\beta \wedge \bar{\theta}^\beta.
 \]
 This implies that
 \[
 {\bar{\omega}^\alpha}\!_\mu \wedge \bar{\psi}^\mu = 0.
 \]
If we define $r_\alpha$ by $dr = r_\alpha \theta^\alpha$, then
\begin{eqnarray*}
d\bar{\psi}^\mu & = & d(r \psi^\mu) = r_\alpha \theta^\alpha \wedge \psi^\mu - r {\psi^\mu}\!_\nu \wedge \psi^\nu\\
& = & -r_\alpha \psi^\mu \wedge \bar{\theta}^\alpha - {\psi^\mu}\!_\nu \wedge \bar{\psi}^\nu.
\end{eqnarray*}
 From these last two equations we see that
 \[
 {\bar{\theta}^\alpha}\!_\beta = {\theta^\alpha}\!_\beta, \quad {\bar{\psi}^\mu}\!_\nu = {\psi^\mu}\!_\nu \quad {\rm and } \quad 
 {\bar{\omega}^\mu}\!_\alpha = r_\alpha \psi^\mu = \frac{r_\alpha}{r} \bar{\psi}^\mu.
 \]
 It then follows that 
 \[
 {\bar{\omega}^\alpha}\!_\mu = -\frac{r^\alpha}{r} \bar{\psi}_\mu
 \]
 since
 \begin{eqnarray*}
 {\bar{\omega}^\alpha}\!_\mu &=& \bar{g}^{\alpha \beta} \bar{\omega}_{\beta \mu} = -\bar{g}^{\alpha \beta} \bar{\omega}_{\mu \beta}
 = - \bar{g}^{\alpha \beta}\bar{g}_{\mu\nu} {\bar{\omega}^\nu}\!_\beta\\
 & = & -\bar{g}^{\alpha \beta}\bar{g}_{\mu\nu} \frac{r_\beta}{r} \bar{\psi}^\nu = - \frac{r^\alpha}{r} \bar{\psi}_\nu,
 \end{eqnarray*}
 keeping in mind that $\bar{g}_{\alpha \mu} =0$.
 
 Before we proceed we discuss some notation.  Let $\bf{}S$ be a symmetric bilinear form on a vector space $V$ with frame $\eee_i$ and dual frame $\omega^i$.  We may write
 \[
 {\bf{S}} = S_{ij} \omega^i \omega^j,\ {\rm with} \ S_{ij} = S_{ji},
 \]
 where we regard this as a the tensor product as opposed to a symmetric product.  We may introduce 1-forms $S_i$  by setting
 \[
 S_i = S_{ij} \omega^j
  \]
 and write 
 \[
 {\bf{S}} = S_i \omega^i.
 \]
 The 1-forms $S_i$ are characterized by the fact that 
 \[
 {\bf{S}} = S_i \omega^i \quad {\rm and}\quad S_i(\eee_j) = S_j(\eee_i)
 \]
 and will be referred to as   ${\bf{S}}$ 1-forms (without explicit mention of the  frame $\eee_i$.) 
 
 We illustrate the use of this notation in two cases for tensor fields on a Semi-Riemannian manifold $M$ with frame field $\eee_i$ and dual frame field $\omega^i$.
 
 First, we consider a smooth function $f : M \to \mathbb{R}$.  The Hessian of $f$, denoted ${\bf{H}}f$, is defined by
 $ {\bf{H}}f = f_{ij}\omega^i \omega^j$, where $f_{ij}$ satisfies the following:
 \[
 f_{ij} \omega^j = df_i - f_j {\omega^j}\!_i  \quad {\rm and}\quad f_{ij} = f_{ji}.
 \]
 Necessarily, 
 \[
 ({\bf{H}}f)_i = f_{ij} \omega^j.
 \]

Second, we consider the Ricci tensor $\bf{R}$ on $M$.  If ${\Omega^i}\!_j$ are the curvature forms associated with the  given frame field, then it is straightforward to show that
\[
R_j = \eee_i \lrcorner {\Omega^i}\!_j.
\]

  We now compute curvatures. 
 \begin{eqnarray*}
{\bar{\Theta}^\alpha}\!_\beta & = & d {\bar{\theta}^\alpha}\!_\beta + {\bar{\theta}^\alpha}\!_\gamma \wedge {\bar{\theta}^\gamma}\!_\beta + {\bar{\omega}^\alpha}\!_\mu \wedge {\bar{\omega}^\mu}\!_\beta\\
& = & {\Theta^\alpha}\!_\beta - \frac{r^\alpha}{r}\bar{\psi}_\mu \wedge \frac{r_\beta}{r} \bar{\psi}^\mu\\
&= & {\Theta^\alpha}\!_\beta.
 \end{eqnarray*}
 \begin{eqnarray*}
 {\bar{\Psi}^\mu}\!_\nu &=& d{\bar{\psi}^\mu}\!_\nu + {\bar{\omega}^\mu}\!_\alpha \wedge {\bar{\omega}^\alpha}\!_\nu + {\bar{\psi}^\mu}\!_\lambda \wedge {\bar{\psi}^\lambda}\!_\nu\\
 &=& {\Psi^\mu}\!_\nu - \frac{r_\alpha r^\alpha}{r^2} \bar{\psi}^\mu \wedge \bar{\psi}_\nu\\
 &=& {\Psi^\mu}\!_\nu - \frac{||\nabla r||^2}{r^2} \bar{\psi}^\mu \wedge \bar{\psi}_\nu.
 \end{eqnarray*}
 \begin{eqnarray*}
 {\bar{\Omega}^\mu}\!_\alpha & = & d{\bar{\omega}^\mu}\!_\alpha + {\bar{\omega}^\mu}\!_\beta \wedge {\bar{\theta}^\beta}\!_\alpha + {\bar{\psi}^\mu}\!_\nu \wedge {\bar{\omega}^\nu}\!_\alpha\\
 &=& d(r_\alpha \psi^\mu) + \frac{r_\beta}{r} \bar{\psi}^\mu \wedge {\theta^\beta}\!_\alpha + {\psi^\mu}\!_\nu \wedge r_\alpha \psi^\nu\\
 &=& dr_\alpha \wedge \psi^\mu + r_\alpha  d\psi^\mu + r_\alpha {\psi^\mu}\!_\nu \wedge \psi^\nu - r_\beta {\theta^\beta}\!_\alpha \wedge \psi^\mu\\
 &=& \frac{1}{r}r_{\alpha\beta} \bar{\theta}^\beta \wedge \bar{\psi}^\mu\\
 & =&  -\bar{\psi}^\mu \wedge     \frac{1}{r} ({\bf{H}} r)_\alpha.
 \end{eqnarray*}
  
   Of course,
 \[
 {\bar{\Omega}^\alpha}\!_\mu = -  \frac{1}{r}({\bf{H}}r)^\alpha \wedge   \bar{\psi}_\mu.
\]

 Next we compute the Ricci forms $\bar{\rho}_\alpha$ and $\bar{\rho}_\mu$ on $X$.  Let $\rho_\alpha$ and $\rho_\mu = \psi_\mu$ denote the Ricci forms on $B$ and $\mathbb{S}^2$, respectively.
 \begin{eqnarray*}
 \bar{\rho}_\alpha &=& \bar{\aaa}_\beta \lrcorner {\bar{\Theta}^\beta}\!_\alpha + \bar{\bbb}_\mu\lrcorner {\bar{\Omega}^\mu}\!_\alpha\\
 &=& \aaa_\beta \lrcorner {\Theta^\beta}\!_\alpha + \bar{\bbb}_\mu \lrcorner \left(\frac{1}{r}r_{\alpha\beta}\bar{\theta}^\beta \wedge \bar{\psi}^\mu \right)\\
 & = & \rho_\alpha - \frac{2}{r}r_{\alpha\beta} \bar{\theta}^\beta\\
 & = &\rho_\alpha - \frac{2}{r} ({\bf{H}}r)_\alpha.
 \end{eqnarray*}
 \begin{eqnarray*}
 \bar{\rho}_\mu & = & \bar{\aaa}_\alpha \lrcorner {\bar{\Omega}^\alpha}\!_\mu + \bar{\bbb}_\nu\lrcorner {\bar{\Psi}^\nu}\!_\mu\\
 &=& \bar{\aaa}_\alpha \lrcorner\left(\frac{1}{r} \bar{\psi}_\mu \wedge ({\bf{H}}r)^\alpha\right) + \bar{\bbb}_\nu \lrcorner \left({\Psi^\nu}\!_\mu - \frac{||\nabla r||^2}{r^2} \bar{\psi}^\nu \wedge \bar{\psi}_\mu\right).\\
 & = & -\frac{1}{r} \aaa_\alpha \lrcorner ({\bf{Hr}})^\alpha \bar{\psi}_\mu + \frac{1}{r} \bbb_\nu \lrcorner {\Psi^\nu}\!_\mu - \frac{||\nabla r||^2}{r^2} (2 \bar{\psi}_\mu - \bar{h}_{\mu\nu}\bar{\psi}^\nu)\\
 &=& \frac{1}{r} \rho_\mu - \left(\frac{\Delta r}{r} + \frac{||\nabla r||^2}{r^2}\right) \bar{\psi}_\mu\\
 & = & \left( \frac{1}{r^2} - \frac{\Delta r}{r} - \frac{||\nabla r||^2}{r^2}\right) \bar{\psi}_\mu.
 \end{eqnarray*}
 
 To finish our computations of the Ricci forms we need to consider the metric on $B$ in more detail.  Thus we set
 \[
 \mathbf{g} = - e^2 dt^2 + g^2 dr^2,
 \]
where we assume are that $e$ and $g$ are smooth and take on only positive values.  Clearly, 
\[
\theta^0 = e dt, \theta_0 = - e dt \ {\rm and \ } \theta^1 = \theta_1 = g dr.
\]

 From the assumptions just made we see that 
 $dr = \frac{1}{g} \theta^1$  which immediately implies that $r_0 = 0$ and $r_1 = \frac{1}{g}$. Thus 
 \[
 ||\nabla r||^2 = \frac{1}{g^2}.
 \]
 In addition,
 \[
 ({\bf{H}}r)_0  = - \frac{1}{g}{\theta^1}\!_0  \quad {\rm and} \quad ({\bf{H}}r)_1 = d\left(\frac{1}{g}\right) = -\frac{1}{g^2} dg.
 \]

We compute ${\theta^0}\!_1$.  Since ${\theta^0}\!_1 = {\theta^1}\!_0$, and
\begin{eqnarray*}
d \theta^0 & = & d(e dt) = e_r dr \wedge dt = -\frac{e_r}{eg} \theta^0 \wedge \theta^1,\\
d \theta^1 & = & d(g dr) = g_t dt \wedge dr = - \frac{g_t}{eg} \theta^1 \wedge \theta^0
\end{eqnarray*}
we see that
\[
{\theta^0}\!_1 = \frac{e_r}{eg} \theta^0 + \frac{g_t}{eg} \theta^1.
\]
Noting that
\begin{eqnarray*}
dg &  = & g_t dt + g_r dr = \frac{g_t}{e} \theta^0 + \frac{g_r}{g} \theta^1\\
de & = & e_t dt + e_r dr = \frac{e_t}{e} \theta^0+ \frac{e_r}{g} \theta^1
\end{eqnarray*}
we obtain
\[
 g_0 = \frac{g_t}{e}, g_1 = \frac{g_r}{g}, e_0 = \frac{e_t}{e} {\rm \ and \ } e_1 = \frac{e_r}{g}.
\]
Hence
\[
{\theta^0}\!_1 = \frac{e_1}{e} \theta^0 + \frac{g_0}{g} \theta^1 = (\ln e)_1\theta^0 + (\ln g)_0 \theta^1.
\]

As a consequence of what we just observed we have the following:
\[
({\bf{H}}r)_0 = -\frac{1}{g} \left( (\ln e)_1\theta^0 + (\ln g)_0 \theta^1 \right) \ {\rm and\ } ({\bf{H}}r)_1 = -\frac{1}{g}\left( (\ln g)_0 \theta^0 + (\ln g)_1 \theta^1\right)
\]
and, moreover,
\[
\Delta r = -({\bf{H}}r)_{00} +({\bf{H}}r)_{11} = \frac{1}{g} \left[ (\ln e)_1 - (\ln g)_1\right] = \frac{1}{g} \left(\ln \frac{e}{g}\right)_1.
\]
In terms of coordinates
\[
\Delta r = \frac{1}{g^2} \left(\ln \frac{e}{g}\right)_r.
\]

Since $\rho_\alpha = K_B \theta_\alpha$, where $K_B$ is the sectional curvature of $B$, we now compute $K_B$.  We do this by computing ${\Theta^0}\!_1$.

\begin{eqnarray*}
{\bar{\Theta}^0}\!_1 &=& d{\bar{\theta}^0}\!_1 = d\left((\ln e)_1 \theta^0 + (\ln g)_0 \theta^1\right)\\
& = & (\ln e)_{11} \theta^1 \wedge \theta^0 -(\ln e)_1 {\theta^0}\!_1 \wedge \theta^1 +(\ln g)_{00} \theta^0 \wedge \theta^1- (\ln g)_0 {\theta^1}\!_0 \wedge \theta^0\\
& = & \left(-(\ln e)_{11}  - (\ln e)_1\:^2  + (\ln g)_{00} + (\ln g)_0\,^2\right)\theta^0 \wedge \theta^1\\
&= & \left(-\frac{e_{11}}{e} +\frac{g_{00}}{g}\right) \theta^0 \wedge \theta_1.
\end{eqnarray*}
Thus
\[
K_B = -\frac{e_{11}}{e} +\frac{g_{00}}{g} = -\frac{1}{eg} \left[ \left(\frac{e_r}{g}\right)_r - \left(\frac{g_t}{e}\right)_t\right].
\]

Substituting formulas for $\rho_\alpha =K_B\theta_\alpha, (Hr)_\alpha, \Delta r$ and $||\nabla r||^2$ into the formulas for the Ricci forms we get the following:
\begin{eqnarray*}
\bar{\rho}_0 &=&  \left(-K_B + \frac{2}{rg} (\ln e)_1\right)  \bar{\theta}^0 + \frac{2}{rg}(\ln g)_0\bar{\theta}^1\\
\bar{\rho}_1 
& = & \frac{2}{rg}(\ln g)_0 \bar{\theta}^0 + \left(K_B + \frac{2}{rg} (\ln g)_1\right) \bar{\theta}^1\\
\bar{\rho}_\mu  &=& \left(\frac{1}{r^2}\left(1- \frac{1}{g^2}\right) -\frac{1}{rg}\left(\ln \frac{e}{g}\right)_1\right) \bar{\psi}_\mu
\end{eqnarray*}

\section{Birkhoff's Theorem}

	Since we have the machinery to prove Birkhoff's Theorem, we will.  As is well known, if ${\bf{T}}=0$, then the Ricci tensor $\bar{\bf{R}}$ on $X$ is trivial. This implies that the Ricci forms must be trivial as well.  This gives us the following four equations:

\begin{eqnarray}
&K_B  =  \frac{2}{rg}(\ln e)_1 \label{eqn:1}\\
&K_B  =  -\frac{2}{rg}(\ln g)_1 \label{eqn:2}\\
&\frac{2}{rg} (\ln g)_0 = 0 \label{eqn:3}\\
& \frac{1}{r^2}\left(1 - \frac{1}{g^2}\right) - \frac{1}{rg}\left(\ln \frac{e}{g}\right)_1 = 0 \label{eqn:4}
\end{eqnarray}

From equation (\ref{eqn:3}) we immediately see that $g= g(r)$, that is, $g$ is a function of $r$ only.  Equations (\ref{eqn:1}) and (\ref{eqn:2}) imply that
$ (\ln eg)_1 = 0$.  Thus, $eg(r) = f(t)$ and we conclude that $e = \frac{f(t)}{g(r)}$, where $f(t) > 0$, for all $t$.  Thus we introduce  a new time variable $\tau$ where $d\tau = f(t) dt$ and the metric on $B$ takes the form
\[
 - \frac{1}{g^2} d\tau^2 + g^2 dr^2.
 \]
Noting that $ e = \frac{1}{g}$ in the new $\tau, r$ coordinates and introducing the function $h = e^2$, equation (\ref{eqn:4}) becomes
\[
\frac{1}{r^2}(1-h) -\frac{h}{r} ( \ln h)_ r = 0,
\]
that is,
\begin{equation}
r h_r + h = 1. \label{eqn:5}
\end{equation}
This is equivalent to $(rh)_r = 1$ and leads immediately to 
\[
h = 1 + \frac{c}{r},
\]
where $c$ is the constant of integration.  

Note however that we have not yet established that equations (\ref{eqn:1}) and (\ref{eqn:2}) are satisfied by our solution
\[
\bar{\mathbf{g}} = -h d\tau^2 + h^{-1} dr^2 + r^2 d\sigma^2,
\]
where we have introduced $d\sigma^2$ for the metric on $\mathbb{S}^2$.

  Since we know that the difference of these two equations is satisfied, it will suffice to show that equation (\ref{eqn:1}) is.  First, we express $K_B$ in terms of $h$.
\[
K_B =  -\left(\frac{e_r}{g}\right)_r = -\frac{1}{2} h_{rr}.
\]
Since
\[
\frac{2}{rg}(\ln e)_1 = \frac{1}{rg^2}(\ln e^2)_r = \frac{1}{r} h_r.
\]
equation (\ref{eqn:1}) may be written
\[
r h_{rr} + 2 h_r = 0.
\]
This follows immediately from equation (\ref{eqn:5}).

Thus under the assumptions made we find the metric has the following representation:
\[
\bar{\mathbf{g}} = -\frac{r +c}{r} d\tau^2 + \frac{r}{r+c} dr^2 + r^2 d\sigma^2.
\]
This is clearly a static metric since $\frac{\partial}{\partial \tau}$ is an irrotational Killing vector field.  If $r_1 = \infty$ then the metric is also  asymptotically flat.  (We will precisely define what we mean by ``asymptotically flat" in Section 6.)  We also note that such a spherically symmetric space-time must have trivial Ricci curvature and hence trivial stress-energy tensor.

If $c >0$  then the metric $\bar{\mathbf{g}}$ is defined for all real $r$ and becomes singular as $r$ approaches $0$.  If $c=0$, $X$ is flat; in fact, it is open submanifold of Minkowski space-time.  If $c < 0$, then one customarily sets $c = -2M$ (where $M >0$) and notes that the metric is defined only for $r > 2M$. In this last case we have obtained the Schwarzchild metric.

We close with a statement of what we have proved.

	{\bf Birkhoff's Theorem}  {\em If a space-time has spherical spatial symmetry and  vanishing Ricci curvature, then is must be static and (assuming $r_1 = \infty$) it must also be asymptotically flat.  If the metic  is singular as $r \to r_0 > 0$ (and $r_1 = \infty$) then the space-time metric is the Schwarzchild metric.}

	\section{Generalizing Birkhoff}
	
	Let $T_SX$ be the vector bundle of 2-dimensional vector spaces tangent to the spatial symmetry spheres and $T_RX$ be the vector bundle that is the orthogonal complement of $T_SX$.  We let $\bar{\bf{g}}_S$ and $\bf{T}_S$ denote the restrictions of $\bar{\bf{g}}$ and $\bf{T}$ to $T_SX$; similarly, let $\bar{\bf{g}}_R$ and $\bf{T}_R$ denote the restriction of $\bar{\bf{g}}$ and $\bf{T}$ to $T_RX$.      Because of the symmetry assumption, there exists a function $\lambda:X \to \mathbb{R}$ (that is the pullback of a function on $B$) such that $\bf{T}_S = \lambda \bar{\bf{g}}_S$.  As an operator this means that $\bf{T}^\#$ (the linear transformation associated to $\bf{T}$ by means of the metric)  has any fiber of  $T_SX$ as an invariant subspace.  Thus, the fibers of  $T_RX$, being the orthogonal complements of the fibers of $T_SM$, must be invariant subspaces of $\bf{T}^\#$ as well.  We will assume that $\bf{T}_R = -\mu \bar{\bf{g}}_R$, where $\mu: X \to \mathbb{R}$ (which also must be the pullback of a function defined on B.)  As a consequence of that assumption, 
\[
\bf{T} = -\mu \bar{\bf{g}}_R + \lambda \bar{\bf{g}}_S.
\]	
This section is devoted to studying space-times which have spherical spatial symmetry and stress-energy tensors of this sort.  The discussion given below mimics the proof of Birkhoff's Theorem.

A straightforward calculation, using Einstein's equation,  implies
	\[
	\bar{\bf{R}} = -8\pi \lambda \bar{\bf{g}}_R  + 8\pi \mu \bar{\bf{g}}_S.
	\]
	This last equation is equivalent to  $\bar{\rho}_\alpha = - 8 \pi \lambda\, \bar{\theta}_\alpha$ and $\bar{\rho}_\mu = 8\pi \mu\, \bar{\psi}_\mu$.  Hence
	\begin{eqnarray} 
	K_B -\frac{2}{rg}(\ln e)_1 = -8\pi\lambda  \qquad \label{eqn:e1}\\
	K_B + \frac{2}{rg} (\ln g)_1 = -8\pi\lambda  \qquad \label{eqn:e2}\\
	\frac{2}{rg}(\ln g)_0 =0  \qquad \quad  \label{eqn:e3}\quad\\
	 \frac{1}{r^2}\left(1 - \frac{1}{g^2}\right) - \frac{1}{rg}\left(\ln \frac{e}{g}\right)_1 = 8\pi\mu \label{eqn:e4}
	\end{eqnarray}

	Half the sum and difference of equations (\ref{eqn:e1}) and (\ref{eqn:e2}) will be more useful. They follow.
	\begin{eqnarray}
	\frac{1}{rg} \left(\ln eg\right)_1 = 0 \qquad  \label{eqn:diff}\\
	\quad K_B +\frac{1}{rg}\left(\ln \frac{g}{e}\right)_1 = - 8 \pi \lambda\   \label{eqn:sum}
	\end{eqnarray}
	
	We begin by introducing the function $h: X \to \mathbb{R}$, where  $h = \frac{1}{g^2}$.   We first observe that equation (\ref{eqn:e3}) is equivalent to asserting $h = h(r)$.  In light of that fact, equation (\ref{eqn:diff}) is equivalent to the following: there  exists a time coordinate $\tau$ so that relative to the $\tau, r$ coordinate system $e^2 = h$ and $g^2 = h^{-1}$, that is, there exists coordinates $\tau, r$ such that 
	\begin{equation}
	\bar{\bf{g}}=-h(r) d\tau^2 + h(r)^{-1} dr^2 + r^2 d\sigma^2. \label{eqn:static}
	\end{equation}	
	Clearly the metric in equation (\ref{eqn:static}) is that of a static space-time.  In fact, it is the form of all the metrics mentioned in the introduction.  
	
	Equation (\ref{eqn:e4}) written in terms of $h$ becomes
\[
\frac{1}{r^2} (1 -h) - \frac{1}{r} h_r = 8\pi\mu,
\]
which can be rewritten as
\begin{equation}
(rh)_r = 1- 8\pi r^2 \mu.  \label{eqn:rh}
\end{equation}
Given that $h = h(r) > 0$, the preceding equation is equivalent to stating that $\mu = \mu(r)$, there exist $c$ and $r_* \in (r_0, r_1)$ such that
\begin{equation}
h(r) = \frac{r+c}{r} - \frac{2}{r} \int_{r_*}^r 4\pi x^2 \mu(x)dx  \label{eqn:h}
\end{equation}
and
\begin{equation}
\int_{r_*}^r 4\pi x^2 \mu(x)dx < \frac{r + c}{2}, \quad {\rm for \ } r \in (r_0,r_1).  \label{eqn:lessthan}
\end{equation}
	
	We only need to consider equation (\ref{eqn:sum}).  It clearly corresponds to a condition on $\mu$ and $\lambda$.  In terms of $h$, equation (\ref{eqn:sum}) becomes
\[
-\frac{1}{2} h_{rr} -\frac{1}{r} h_r  = -8\pi\lambda,
\]
which can be rewritten

\[
(rh)_{rr}  = 16 \pi r \lambda.
\]
Because of equation (\ref{eqn:rh}), this last equation is equivalent to 
\begin{equation}
\lambda = -\mu - \frac{r}{2} \mu_r. \label{eqn:tov}
\end{equation}

As a consequence of the preceding discussion we have our first generalization of Birkhoff's Theorem.

\begin{thm}
 If  $X$ is a spherically symmetric space-time with metric $\bar{\bf{g}}$ and stress-energy tensor   $\bf{T} = -\mu \bar{\bf{g}}_R + \lambda \bar{\bf{g}}_S$, then  $\mu = \mu(r)$, there exist $c$ and $r_* \in (r_0, r_1)$ such that  $\mu$ satisfies inequality {\rm (\ref{eqn:lessthan})}, equation {\rm (\ref{eqn:tov})} holds and there exists a time coordinate $\tau$ so that
  the metric $\bar{\mathbf{g}}$ is given by equation {\rm (\ref{eqn:static})} with $h$ defined by {\rm (\ref{eqn:h})}.  Conversely, if $X$ is a space-time with metric ${\rm \bar{\bf{g}}}$ given by equation {\rm (\ref{eqn:static})} where $h$ is defined by equation {\rm (\ref{eqn:h})}, with $\mu = \mu(r)$ necessarily satisfying inequality {\rm (\ref{eqn:lessthan})} and $\lambda$ is given by equation {\rm (\ref{eqn:tov})}, then $X$ with metric $\bar{\bf{g}}$ is spherically symmetric and has stress-energy tensor   $\bf{T} = -\mu \bar{\bf{g}}_R + \lambda \bar{\bf{g}}_S$.

   \end{thm}

Obviously, a metric $\bar{\bf{g}}$ given by equation (\ref{eqn:static}) where $h$ satisfies equation (\ref{eqn:h})  is static and   any $\lambda$  given by  equation (\ref{eqn:tov}) satisfies $\lambda =  \lambda(r)$ when $\mu = \mu(r)$.  We see we have such metrics for every $\mu = 
\mu(r)$ which satisfies (\ref{eqn:lessthan}).  We therefore consider some examples.

\begin{exam}\end{exam} 
\noindent
 Suppose $\bf{T}$ is a multiple of the metric, that is, $\lambda = -\mu$ throughout $X$.  It follows from Theorem 1 that $\mu$ and $\lambda$ are constants.  If we suppose that $r_0 = 0$, $c=0$ and $r_* = 0$ (which does not present any problems given the form of the integrand)   in order  that the metric extends to the world line of the center of symmetry, then inequality (\ref{eqn:lessthan}) implies $r_1 < \sqrt{\frac{3}{8\pi\mu}}$.  We get a de-Sitter metric if $\mu > 0$, a Minkowski space if $\mu = 0$ and an anti-de Sitter space if $\mu <0$.
\begin{exam}\end{exam} 
\noindent
Suppose $r_0 = 0, r_1 = \infty$ and  assume $\rho$ and $b$ are positive constants.  Set
\[ 
\mu = \rho \ {\rm for\ } 0 < r < b < \sqrt{\frac{3}{8 \pi \rho}} \quad {\rm and} \quad \mu = 0 \ {\rm for \ } r \geq b.
\]
We choose $c = r_* = 0$, and then  note $r_1 = \infty$ is allowed in this situation since inequality (\ref{eqn:lessthan}) is satisfied for all $r > 0$.  In addition, by equation (\ref{eqn:tov})
\[
\lambda = -\mu \ {\rm for\ } r \neq b \quad {\rm and } \quad \lambda(b) =  \frac{b \rho}{2} \delta_{r-b}.
\]
Introducing $m(r) = \frac{4}{3} \pi r^3 \rho$, which can be regarded as some sort of mass of the body bounded by the 2-sphere of curvature $\frac{1}{r^2}$ with density $\rho$, then
\[
h(r) = 1 - \frac{2 m(r)}{r} \ {\rm for\ } r< b \quad {\rm and } \quad h(r) = 1 - \frac{2m(b)}{r} \ {\rm for \ } r \geq b.
\]
This metric is the static approximation\cite{cw2018} of the simplest GEneric Object of Dark Energy (GEODE).

\begin{exam}\end{exam} 
\noindent
Suppose  $\mu$ is constant and  $r_0 = r_* > 0$.  Let $c - \frac{8 \pi \mu}{3} r_0^3 = R_S$, which we regard as the Schwarzchild radius, and $\Lambda = 8\pi \mu$, which we regard as the cosmological constant.  We get the Kottler-Weyl-Trefftz metric for which
\[
h(r) = 1- \frac{R_S}{r} - \frac{\Lambda r^2}{3}.
\]

\section{Further Generalization}

We have been using special frames throughout this paper and at this point we want to give those frames a name.  Recall that $r$ and the direction of $\frac{\partial}{\partial t}$ have geometric meaning and thus physical meaning.
\begin{defin}  We call an orthonormal frame a Birkhoff frame if the time-like vector of the frame is a multiple of $\frac{\partial}{\partial t}$ and the first spacial vector of the frame is a multiple of $\frac{\partial}{\partial r}$.  The coframe dual to this frame is referred to as a Birkhoff coframe.
\end{defin}

For a spherically symmetric space-time, we have shown that
\[
{\bf{T}} =  -\mu_0 \bar{\theta}_0 \bar{\theta}^0 + \kappa (\bar{\theta}_1\bar{\theta}^0 - \bar{\theta}_0 \bar{\theta}^1) + \mu_1 \bar{\theta}_1 \bar{\theta}^1 + \lambda \bar{\psi}_\nu \bar{\psi}^\nu
\]
with respect to a Birkhoff coframe.  The components $\mu_\alpha, \lambda$ and $\kappa$ are pullbacks of functions on $B$, i.e., they depend only on $t$ and $r$.  We will understand that unless specifically restricted to be functions  of $r$ these components are to be regarded as functions of $t$ and $r.$

It follows from Einstein's equation that  $\bar{\bf{R}}$ has the same eigenvectors with respect to $\bar{\bf{g}}$ as $\bf{T}$ does.  By examining the Ricci forms of a spacially symmetric space-time we see that  $\bar{\bf{R}}(\bar{\aaa}_0. \bar{\aaa}_1) = \frac{2}{rg}(\ln g)_0$, where  $\bar{\aaa}_0,  \bar{\aaa}_1$ are the first two vectors of a Birkhoff frame.    If the metric on $X$ is spherically symmetric and static, then $\bar{\bf{R}}(\bar{\aaa}_0. \bar{\aaa}_1) = 0$ and consequently $\bar{\aaa}_0$ and $ \bar{\aaa}_1$ are eigenvectors of $\bar{\bf{R}}$ with respect to $\bar{\bf{g}}$. Hence, $\bar{\aaa}_0$ and $ \bar{\aaa}_1$ are eigenvectors of $\bf{T}$ with respect to $\bar{\bf{g}}$ as well.    Obviously, the functions $\mu_\alpha$ and $\lambda$ depend only on $r$ and $\kappa = 0$.  The following has been shown.  \\

\noindent{\bf{Observation.}}
\emph{If $X$ is a static spherically symmetric space-time, then ${\bf{T}}$ is diagonalized by a Birkhoff frame  with components $\mu_\alpha = \mu_\alpha(r)$ and $\lambda = \lambda(r)$, i.e., the stress-energy tensor }
\begin{equation}
{\bf{T}} = -\mu_0(r) \bar{\theta}_0 \bar{\theta}^0 + \mu_1(r) \bar{\theta}_1 \bar{\theta}^1 + \lambda(r) \bar{\psi}_\nu \bar{\psi}^\nu.
\label{eqn:T}
\end{equation}

Based on the Observation, a reasonable starting point for any further generalizations would require that $\bf{T}$ satisfy equation (\ref{eqn:T}).  However, are we going to demand anything of the $\mu_\alpha$ and $\lambda$?  We will only impose a condition on $\mu_1$.  In fact, our starting point for the derivation of our main theorem is the imposition of the following conditions:
\begin{eqnarray*}
&(i)&\quad X \ {\rm is \ a\ spherically\ symmetric \ space\!-\!time}.\\
&(ii)&\quad \bf{T} {\rm \ has\ the\ form\ given \ in \ equation\  (\ref{eqn:T}}).\\
&(iii)&\quad  \mu_1 = \mu_1(r).
\end{eqnarray*}

Because of $(i)$ and $(ii)$, Einstein's equation yields the following:
\begin{eqnarray}
&\frac{1}{rg}\left(\ln eg\right)_1  = 4\pi(\mu_0 + \mu_1)  \label{eqn:E1}\\
&K_B + \frac{1}{rg} \left(\ln \frac{g}{e}\right)_1 = -8\pi\lambda \label{eqn:E2}\\
&\frac{2}{rg} (\ln g)_0 = 0    \label{eqn:E3}\\
& \frac{1}{r^2} \left(1 - \frac{1}{g^2}\right) - \frac{1}{rg}\left(\ln \frac{e}{g}\right)_1 = 4\pi(\mu_0 - \mu_1) \label{eqn:E4}
\end{eqnarray}

As before, we let $h = \frac{1}{g^2}$ but also introduce $p = eg$.  (The function $p$ stands for product not pressure.)
Just as we noticed in earlier sections, equation (\ref{eqn:E3}) is equivalent to asserting $h = h(r)$.  Rewriting equation (\ref{eqn:E1}) in terms of $h$ and $p$ gives
\begin{equation}
h \left(\ln p\right)_r =4\pi r(\mu_0 + \mu_1). \label{eqn:p}
\end{equation}
Moreover, rewriting  (\ref{eqn:E4}) using $h$ and $p$, we get the following:
\begin{eqnarray*}
\frac{1}{r^2} (1-h) - \frac{h}{r} (\ln hp)_r =  4\pi (\mu_0 - \mu_1)\\
\frac{1}{r^2}(1-h)  - \frac{1}{r}{h_r} - \frac{h}{r}(\ln p)_r =4\pi (\mu_0 - \mu_1)
\end{eqnarray*}
Using equation (\ref{eqn:p}) to  substitute into the third term, this last equation becomes after some simplification
\begin{equation}
(r h)_r = 1 - 8\pi r^2 \mu_0. \label{eqn:rhagain}
\end{equation}
The preceding statement is equivalent to asserting that $\mu_0 =  \mu_0(r)$, there exists $c$ and $r_* \in (r_0,r_1)$ such that
\begin{equation}
h(r) = \frac{r+c}{r} - \frac{2}{r}\int_{r_*}^r 4\pi x^2 \mu_0(x)dx. \label{eqn:hagain}
\end{equation}
and
\begin{equation}
\int_{r_*}^r 4\pi x^2 \mu_0(x)dx < \frac{r + c}{2}, \quad {\rm for \ } r \in (r_0,r_1).  \label{eqn:lessthanagain}
\end{equation}
Taking into account $(iii)$, we can now solve equation (\ref{eqn:p}) for $p$.  We find that there exists a function of $t$, $f(t) >0$, so that
\begin{equation}
p = f(t) \exp\left[\int_{r_*}^r \frac{4\pi x^2 (\mu_0(x) + \mu_1(x))}{x+ c - 2 \int_{r_*}^x 4\pi y^2 \mu_0(y)dy} dx\right]. \label{eqn:pagain}
\end{equation}
We introduce a new time variable $\tau$ so that $f(t) dt = d\tau$ and note that in $\tau, r$ coordinates $p = p(r)$.  
The function $p$ is then given by equation (\ref{eqn:pagain}) with $f(t)$ removed.    The metric has the following form:
\begin{equation}
\bar{\bf{g}} = -p^2 h d\tau^2 + h^{-1}dr^2 + r^2 d\sigma^2.  \label{eqn:gagain}
\end{equation}

What we have observed so far is that, because of $(iii)$, equations (\ref{eqn:E1}),  (\ref{eqn:E3}) and (\ref{eqn:E4}) are equivalent to asserting the following:  $\mu_0 = \mu_0(r)$, which is subject to equation (\ref{eqn:lessthanagain}), and there exists coordinates $\tau,r$ so that $\bar{\mathbf{g}}$ is given by equation (\ref{eqn:gagain}) with $h$ and $p$ given by equations (\ref{eqn:hagain}) and (\ref{eqn:pagain}), with $f(t)$ removed, respectively.

We still have to deal with  equation  (\ref{eqn:E2}); this should correspond to a  condition on $\mu_0, \mu_1$ and $\lambda$.    Written in terms of $h$ and $p$, equation (\ref{eqn:E2}) becomes
\[
\frac{1}{2p} \left( \frac{(p^2 h)_r}{p}\right)_r + \frac{1}{rp}\left(ph\right)_r = 8 \pi \lambda.
\]
We note this implies that $\lambda = \lambda(r)$ and proceed with a calculation. After carrying out the indicated differentiations, we get
\[
(\ln p)^2_r h + (\ln p)_{rr} h + (\ln  p)_r h_r + \frac{1}{2}(\ln p)_r h_r + \frac{1}{2} h_{rr} + \frac{1}{r}h_r + \frac{1}{r} (\ln p)_r h = 8 \pi \lambda.
\]
Regrouping, we can write
\[
(rh)_{rr} + 2r[(\ln p)_r h]_r + (\ln p)_r h + (\ln p)_r [ 2r(\ln p)_r h + (rh_r)_r] = 16\pi r \lambda.
\]
Using equations (\ref{eqn:p}) and (\ref{eqn:rhagain}), we get
\begin{equation}
\lambda = \frac{1}{4}(3\mu_1 -  \mu_0) + \frac{r}{2}(\mu_1)_r + \frac{(\mu_0 + \mu_1)(1 + 8\pi r^2 \mu_1)}{4h}. \label{eqn:TOV}
\end{equation}
(It easily follows that if $\mu_1 = -\mu_0$ then this equation reduces to equation (\ref{eqn:tov}).)

The forgoing has proved the following generalization of Birkhoff's theorem.
\begin{thm}If $X$ is a spherically symmetric space-time with 
\[
{\bf{T}} = -\mu_0 \bar{\theta}_0 \bar{\theta}^0 + \mu_1 \bar{\theta}_1 \bar{\theta}^1 + \lambda \bar{\psi}_\nu \bar{\psi}^\nu
\] 
and $\mu_1 = \mu_1(r)$, then  $\mu_0 =\mu_0(r)$ and there exist $c$ and $r_* \in (r_0,r_1)$ such that $\mu_0$ satisfies inequality {\rm (\ref{eqn:lessthanagain})}, equation {\rm (\ref{eqn:TOV})} holds and there exists a time coordinate $\tau$ such that the metric $\bar{\bf{g}}$ is given by equation {\rm (\ref{eqn:gagain})}, where $h$ is given by equation {\rm (\ref{eqn:hagain})} and
\begin{equation}
p =\exp\left[\int_{r_*}^r \frac{4\pi x^2 (\mu_0(x) + \mu_1(x))}{x+ c - 2 \int_{r_*}^x 4\pi y^2 \mu_0(y)dy} dx\right]. \label{eqn:P}
\end{equation}
Conversely,  if $X$ is a space-time with metric $\bar{\bf{g}}$ given by equation {\rm (\ref{eqn:gagain})} where $h$ is given by {\rm (\ref{eqn:hagain})} and $p$ is given by equation {\rm (\ref{eqn:P})}, with $\mu_\alpha = \mu_\alpha(r)$, and $\mu_0$ necessarily satisfying equation {\rm (\ref{eqn:lessthanagain})}, and $\lambda$ given by equation {\rm (\ref{eqn:TOV})}, then $X$ is a spherically symmetric space-time and has stress-energy tensor 
\[
{\bf{T}} = -\mu_0 \bar{\theta}_0 \bar{\theta}^0 + \mu_1 \bar{\theta}_1 \bar{\theta}^1 + \lambda \bar{\psi}_\nu \bar{\psi}^\nu.
\] 

\end{thm}

It is clear from Observation that the metrics described in Theorem 2 account for all static spherically symmetric space-times.
We now consider some examples.

\begin{exam}\end{exam} The purpose of this example is to present a non-trivial closed form solution satisfying the conditions given in Theorem 2. Suppose there exists constants $m_0$ and $m_1$ such that 
\(
\mu_\alpha = \frac{m_\alpha}{r^2}, \ {\rm for\ } r_0 < r < r_1.
\)
Because of the form of the integrands with which we are dealing, we may take $r_* = r_0$.  The calculations are straightforward; we get the following, where $C = c + 8\pi m_0 r_0$.
\[ 
h(r) = 1 - 8\pi m_0 +\frac{C}{r} \quad {\rm and} \quad p = \left(\frac{h(r)}{h(r_0)}\right)^{\frac{4\pi(m_0 + m_1)}{1 - 8\pi m_0}}
\]
so that
\[
\bar{\bf{g}} =-h(r_0)^{\frac{8\pi(m_0 + m_1)}{8\pi m_0-1}}
 h(r)^{\frac{1 - 8\pi m_1}{1-8 \pi m_0 }}d\tau^2 + \frac{1}{h(r)} dr^2 + r^2 d\sigma^2
\]
and 
\[
\lambda = \frac{(m_0 + m_1)[8\pi(m_0+ m_1) r - C]}{4[(1 - 8\pi m_0)r + C]}\frac{1}{r^2}.
\]
For physical reasons we want $m_0 \geq 0$.  To get a Lorentzian metric, with $\frac{\partial}{\partial \tau}$ time-like, one of the following will do:
\begin{eqnarray*}
&{\rm If \ } m_0 < \frac{1}{8\pi}\ {\rm and \ } C \geq 0, \ {\rm then\  } r_0 \ {\and \ } r_1 \ {\rm are \ unrestricted}.\\
&{\rm If\ } m_0 < \frac{1}{8\pi} \ {\rm and\ } C < 0, \ {\rm then\ } r_0 \geq \frac{C}{8\pi m_0 -1}.\\
&{\rm If \ } m_0 > \frac{1}{8\pi}, \ {\rm then \ } C >0 \ {\rm and \ } r_1 \leq  \frac{C}{8\pi m_0 -1}.
\end{eqnarray*}

An interesting case occurs if $m_0 + m_1 =0$. Then 
\[
\bar{\bf{g}} =  - h(r) d\tau^2 + h(r)^{-1}dr^2 + r^2 d\sigma^2\quad {\rm and }\quad \lambda = 0.
\]
This is a collection of metrics which includes the Schwartzchild metric, which occurs when $m_0 = 0$, that have no transversal pressure with respect to the Birkhoff frame.

\begin{exam}\end{exam} We suppose all of the following: $r_0 = c = r_* = 0, r_1 = \infty$ and $\rho$ is a non-negative valued function on $(0,\infty)$.  We define $m(r)$ and place an additional restriction on $\rho$ by asserting
\[
m(r) = \int_0^r 4\pi x^2 \rho(x) dx < \frac{r}{2}, \ {\rm for\  all\ } r >0.
\]
Moreover,  suppose that $\bf{T}$ is the stress-energy tensor of a perfect fluid and set $\rho = \mu_0 \geq 0$ and $ P = \mu_1 = \lambda$.  Substituting into equation (\ref{eqn:TOV}) yields the TOV equation
\[
P_r   = -(P + \rho)\frac{m(r) + 4\pi r^3P}{r[r-2m(r)]}.
\]
In order for the metric to exist at $r = 0$ it is necessary and sufficient that
\[
\left |\int_0^r x(\rho(x)  + P(x)) dx\right | < \infty, \ {\rm for\ any\ } r >0.
\]

\noindent
{\bf Remark.}  If we write $h(r) = 1 - \frac{2m(r)}{r}$ so that $m(r) = c + \int_{r_*}^r 4\pi x^2 \mu_0(x) dx$, then equation (\ref{eqn:TOV}) can be written as follows:
\[
\lambda - \mu_1 = \frac{r}{2} \left[(\mu_1)_r + \frac{(\mu_0 + \mu_1) (m(r) + 4\pi r^3 \mu_1)}{r[r- 2m(r)]}\right].
\]
We notice its close relation to TOV; hence we call it or equation (\ref{eqn:TOV}) the  generalized TOV equation.  

We end this section by applying Theorem 2 to the study of stress-energy tensors $\bf{T}$ which are ``close to $\bf{0}$."
What we mean by that will be made clear in what is to follow.

\begin{defin}  Let $\bf{B}: V \times V \to \mathbb{R}$ be a symmetric bilinear form on the vector space $V$.  The null space of $\bf{B}$, $N(\bf{B})$, is defined as follows:
\[ 
N(\bf{B}) = \{\bf{w} \in V: \bf{B}(\bf{w}, \bf{v}) = 0, \ {\rm for \ all \ } \bf{v} \in V\}.
\]
\end{defin}

The next lemma is a straightforward result from linear algebra.
\begin{lem}
Suppose  $V$ is an $n$-dimensional vector space with inner product.  If $\dim (N({\bf{B}})) = n-1$, then the 1-dimensional subspace of $V$ orthogonal to $N({\bf{B}})$ is the eigenspace of $\bf{B}$ with respect to that metric associated to the nonzero eigenvalue of $\bf{B}$ with respect to that metric.

\end{lem}

\begin{cor}Let $X$ be a spherically symmetric space-time such that $N({\bf{T}}_x)$ contains the spacial vectors of the Birkhoff frame, for every $x \in X$, then  $\bf{T} = 0$
\end{cor}

\begin{proof}
Suppose the hypothesis are true.  From the definition of the Birkhoff frame and Lemma 1, we see that
\[
\bf{T} = - \mu_0 \bar{\theta}_0 \bar{\theta}^0.  
\]
All of the conditions $(i), (ii)$ and $(iii)$ are satisfied since $\mu_1$ is a constant function. Hence, $\mu_0 = \mu_0(r)$ and (\ref{eqn:TOV}) holds.  Substituting $0$ for $\mu_1$ and $\lambda$ in equation (\ref{eqn:TOV}),  we get
\[
-\frac{\mu_0}{4}+ \frac{\mu_0}{4h} = 0.
\]
Thus, on any open interval where $\mu_0 \neq 0$, we must have $h = 1$; that is
\[
c = 2 \int_{r_*}^r 4\pi x^2 \mu_0(x) dx.
\]
This is only possible if $\mu_0 = 0$.
\end{proof}

As a consequence of this corollary we have the following.
\begin{cor}
Birkhoff's Theorem is still true for a spherically symmetric space-time $X$ if we replace the condition that the space-time is a vacuum by the condition that $N({\bf{T}}_x)$ contains the spacial vectors of the Birkhoff frame, for every $x \in X$.
\end{cor}

We have yet one more corollary, but first a definition.
\begin{defin}
Let $X$ be a space-time.  We say that $X$ is dust, if $N({\bf{T}}_x)$ is the orthogonal complement of a time-like vector, for all $x \in X$.
\end{defin}

An immediate consequence of the Observation and Corollary 1 is the well-known result.

\begin{cor}
Suppose $X$ is a spherically symmetric space time.  If $X$ is dust, then $X$ cannot be static.
\end{cor}

\section{Asymptotics}

For the space-times we are considering where the coordinate $r$ has geometric meaning, we have the following definition.\\

\noindent
{\bf Definition of Asymptotically Flat.}  {\em A space-time $X$ which is spherically symmetric is said to be asymptotically flat if the metric $\bar{\bf{g}}$ approaches a limiting value and  the Riemann curvature tensor of $\bar{\bf{g}}$ approaches {\rm 0} as $r$ approaches $r_1$.}

These conditions are certainly met by the metrics that arise in Birkhoff's Theorem.

Let's elaborate on the conditions specified in the definition.  That $\bar{\bf{g}}$ has a limit as $r \to r_1$ means that $\lim_{r \to r_1} e^2$ and $\lim_{r \to r_1} g^2$ are positive reals.  To say that the Riemann curvature tensor approaches 0 means that all of its components with respect to some orthonormal moving frame approach 0;  thus the Riemann curvature tensor approaches 0 as $r \to r_1$  if and only if
\[
{\bar{\Theta}^0}\!_1 \to 0, \quad {\bar{\Psi}}^2\!_3 \to 0, \quad {\bar{\Omega}^\alpha}\!_\nu \to 0 \quad{\rm as}\quad r \to r_1.
\]
From the formulas for these curvature forms we see that these last limits take the following form:
\[
K_B \to 0, \  \frac{1}{r^2}\left(1 - \frac{1}{g^2}\right) \to 0, \  \frac{1}{rg^2}(\ln e)_r \to 0, \  \frac{1}{rg^2}(\ln g)_r \to 0 \ {\rm as}\ r \to r_1.
\]

We will now characterize those static spherically symmetric space-times   that are asymptotically flat.  However, we will do so under the physically reasonable assumption that $\mu_0 \geq 0$.
\begin{thm}
Assume $X$ is a static space-time with spherical spatial symmetry and 
\[
{\bf{T}} = -\mu_0 \bar{\theta}_0 \bar{\theta}^0 + \mu_1 \bar{\theta}_1 \bar{\theta}^1 + \lambda \bar{\psi}_\nu \bar{\psi}^\nu,
\]
where $\mu_0 \geq 0$.  Then $X$ is asymptotically flat if and only if one of the following occurs:
\begin{eqnarray*} 
&{\rm (i)}&\  r_1 < \infty, \ h(r) = 1 +\frac{2}{r}\int_r^{r_1} 4 \pi x^2 \mu_0(x) dx,\ {\rm and}\\
 \qquad  & \qquad&\qquad\qquad\lim_{r \to r_1} \mu_0 = \lim_{r \to r_1} \mu_1 = \lim_{r \to r_1} (\mu_1)_r =0.\\
&{\rm (ii)}&\  r_1  = \infty,  \lim_{r \to \infty} r^2 \mu_0(r) = l, \ {\rm where\ }  0\leq l < \frac{1}{8\pi}, \ {\rm and}\\
\qquad & \qquad& \qquad \qquad   \int_{r_*}^\infty x (\mu_0(x) + \mu_1(x)) dx \  {\rm exists} .
\end{eqnarray*}
\end{thm}
 \begin{proof}
First, let's write the four limits given above in terms of $h$, $\mu$ andf $\lambda$.  We can do this by taking advantage of Einstein's equation, more explicitly, equations (\ref{eqn:E1}), (\ref{eqn:E2}) and (\ref{eqn:E4}).  We get that the above limits hold, and thus the Riemann curvature tensor approaches 0 as $r \to r_1$, if and only if the following do:
\begin{equation}
\lim_{r \to r_1} \frac{1-h}{r^2} =0, \  \lim_{r \to r_1} \mu_0 = 0, \  \lim_{r \to r_1} \mu_1 =  0,\ \lim_{r \to r_1} \lambda = 0  \label{eqn:lim1}
\end{equation}
Second, as a consequence of equation (\ref{eqn:TOV}), 
\begin{equation}
\lim_{r \to r_1} {\bf{T}} = 0  \Leftrightarrow \lim_{r \to r_1}\mu_\alpha = \lim_{r \to r_1}\lambda =0 \Leftrightarrow \lim_{r \to r_1} \mu_\alpha = \lim_{r \to r_1} (\mu_1)_r = 0. \label{eqn:lim2}
\end{equation}

Assume $X$ is asymptotically flat.  Suppose $r_1 < \infty$.  Since the limit of the metric exists,  $\lim_{r \to r_1} h$ is positive real. 
As a consequence of the first limit appearing in (\ref{eqn:lim1}),  $r_1 < \infty$ implies $\lim_{r\to r_1} h(r) = 1$.  Using equation (\ref{eqn:hagain}), this implies 
\[
c = 2 \int_{r_*}^{r_1} 4\pi x^2 \mu_0(x) dx.
\]
Consequently
\[
h(r) = 1 + \frac{2}{r}\int_r^{r_1} 4\pi x^2 \mu_0(x) dx.
\]
Since the Riemann curvature tensor approaches 0 as $r \to r_1$, it follows, from (\ref{eqn:lim1}) and (\ref{eqn:lim2}), that the remainder to the conditions in (i) hold.

Now suppose that $r_1 = \infty$. 
Thus, it must be true that $\lim_{r \to \infty} h  =n$, where $n >0$.
According to  equation (\ref{eqn:hagain}) this implies
\[
1- 8\pi \lim_{r \to\infty} \frac{1}{r} \int_{r_*}^r x^2 \mu_0(x) dx = n.
\]
If we let 
\[
l = \lim_{r \to \infty} \frac{1}{r}\int_{t_*}^r x^2 \mu_0(x) dx,
\]
then $0 \leq l < \frac{1}{8\pi}$, since we have assumed that $\mu_0 \geq 0$.
The  integral appearing in the above limit is an nondecreasing function of $r$.  Thus  the integral is bounded for all $r$ or approaches $\infty$ as $r \to \infty$.  If it is bounded for all $r$ then  $\lim_{r \to \infty} r^2 \mu_0(r) = 0$.
If the integral approaches $\infty$ as $r \to \infty$, then we apply l'Hospital's Rule to the limit and obtain
$\lim_{r \to  \infty} r^2\mu_0(r) = l$.  Either way, $\lim_{r \to \infty} r^2 \mu_0(r) = l$, where $0 \leq l < \frac{1}{8\pi}$.

It must also be true that $\lim_{r \to \infty} p^2 h$ is a positive  real. Since the $\lim_{r \to \infty} h$ exists and $p$ is an exponential 
\[
\lim_{r \to \infty} \ln p = \lim_{r \to \infty} \int_{r_*}^r \frac{4\pi x(\mu_0(x) + \mu_1(x))}{h} dx \ {\rm exists}.
\]
Again, since $\lim_{r \to \infty} h = n$, $\int_{r_*}^\infty x(\mu_0(x) + \mu_1(x)) dx$ must exist. 

Now assume that (i) holds. Clearly $\lim_{r \to r_1}h = 1$.  Thus the first limit in (\ref{eqn:lim1}) holds.  By (\ref{eqn:lim2}) the remainder of the limits in (\ref{eqn:lim1}) holds.  Since the limits of $h, \mu_0$ and $\mu_1$ all exist as $r \to r_1$,  the integral in the definition of $p$ is a proper integral even when the upper limit is $r_1$. Hence, $\lim_{r \to r_1} p^2 h$ exists.  Indeed, $X$ is asymptotically flat.

Finally, assume what is asserted in (ii) is true.  It is not hard to show that $\lim_{r \to \infty} h(r)$ is a positive real from the assertion that $\lim_{r \to \infty} r^2 \mu_0(r) = l$, where $0 \leq l < \frac{1}{8 \pi}$.  It also follows from this limit that the second limit in (\ref{eqn:lim1}) holds.  Since $\int_{r_*}^\infty x(\mu_0(x) + \mu_1(x) )dx$ exists and $\lim_{x \to \infty}h$ exists, $\lim_{r \to \infty} \ln p$ exists and thus $\lim_{r \to \infty} p^2 h$ is a positive real.   Moveover, that  improper integral given in (ii) exists implies that $\lim_{r \to \infty} r(\mu_0(r) + \mu_1(r)) =0$.  But the existence of the first limit in (ii) implies $\lim_{r \to \infty} r \mu_0(r) = 0$.  Thus $\lim_{r \to \infty} r \mu_1(r) = 0$.  Hence, the third limit in (\ref{eqn:lim1}) holds.   However, we can apply l'Hospital's Rule to $\lim_{r \to \infty} \frac{\mu_1(r)}{r^{-1}} = 0$ to show that  $\lim_{r \to \infty} (\mu_1)_r = 0$.  By (\ref{eqn:lim2}), the fourth limit in (\ref{eqn:lim1}) holds.  Thus $X$ is asymptotically flat. 
\end{proof}

We consider the examples presented above.  In Example 1, only the Minkowski space-time is asymptotically flat.  All the metrics in Example 2 are asymptoticaly flat.  The only space-times in Example 4 that is asymptotically flat are those for which $m_0 + m_1 =0$.

For the space-times in  Example 5 to be asymptotically flat we must require that $\lim_{r \to \infty} r^2 \rho(r) < \frac{1}{8\pi}$.  (The assumption that  $m(r) <\frac{r}{2}$, for all  $r >0$,  would only imply $\lim_{r \to \infty} r ^2\rho(r) \leq \frac{1}{8\pi}$.)   We must also require that $\int_{r_*}^\infty x(\rho(x) + P(x)) dx$ exist.

\section{Acknowledgement}
I would like to thank Kevin Croker for his assistance  revising the introduction to this paper and his suggestions regarding the readability of the paper in general.

\end{document}